\documentclass[12pt,reqno]{amsart}

\usepackage{amsfonts,color,amsthm,amsmath,amssymb}

\def\Box{\vcenter{\vbox{\hrule\hbox{\vrule
     \vbox to 8.8pt{\hbox to 10pt{}\vfill}\vrule}\hrule}}}

\newcommand{\tr}{\textup{Tr}}
\newcommand{\al}{\alpha}
\newcommand{\eN}{{\mathcal N}}
\newcommand{\eL}{{\mathcal L}}
\newcommand{\eH}{{\mathcal H}}
\newcommand{\C}{{\mathbb C}}
\newcommand{\F}{{\mathbb F}}
\newcommand{\ra}{\rangle}
\newtheorem{thm}{Theorem}

\newtheorem{lemma}[thm]{Lemma}

\newtheorem{definition}[thm]{Definition}

\numberwithin{equation}{section}

\begin{document}

\title{Quantum channels from association schemes}


\author{Tao Feng}
\address{Zhejiang University}
\curraddr{}
\email{tfeng@zju.edu.cn}
\thanks{TF is supported in part by the Fundamental Research Funds for the Central
Universities of China, Zhejiang Provincial Natural Science Foundation (LQ12A01019), National Natural Science
Foundation of China (11201418).}

\author{Simone Severini}
\address{University College London}
\curraddr{}
\email{simoseve@gmail.com}
\thanks{SS is supported by the Royal Society.}



\date{}

\dedicatory{}

\begin{abstract}
We propose in this note the study of quantum channels from association schemes. This is done by interpreting the $(0,1)$-matrices of a scheme as the Kraus operators of a channel. Working in the framework of one-shot zero-error information theory, we give bounds and closed formulas for various independence numbers of the relative non-commutative (confusability) graphs, or, equivalently, graphical operator systems. We use pseudocyclic association schemes as an example. In this case, we show that the unitary entanglement-assisted independence number grows at least quadratically faster, with respect to matrix size, than the independence number. The latter parameter was introduced by Beigi and Shor as a generalization of the one-shot Shannon capacity, in analogy with the corresponding graph-theoretic notion.
\end{abstract}

\maketitle

\section{Introduction}

Association schemes have traditionally an important role in information theory and coding \cite{Delsarte98}. This is a short paper in which we introduce a class of quantum channels based on association schemes. We study the first properties of these channels concerned with zero-error capacities \cite{Korner89}. It has been recently shown that the Shannon capacity of a graph can be improved if sender and receiver share a certain type of quantum state \cite{Leung10}. This newly introduced capacity, which is called entanglement-assisted capacity, is conjectured to be equal to the Lov\'{a}sz $\vartheta$-function for classical channels; a quantum version of the $\vartheta$-function upper bounds the analogue capacity for quantum channels \cite{Duan12}. The largest number of messages that can be distinguished without error down a quantum (or even a classical) channel corresponds to the dimension of a certain subspace, whose properties are completely determined by the channel map. The map defines an operator system \cite{Paul03} as a generalization of the confusability graph in classical Shannon zero-error information theory. The system records which pairs of inputs can lead to the same output. The system, also called non-commutative graph in \cite{Duan12}, is a well-defined generalization of graphs, since it encodes faithfully a graph whenever the Kraus operators represent the edges. By interpreting the $(0,1)$-matrices of an association scheme as the Kraus operators, we show that these notions of quantum information theory extends to association schemes. We work with pseudocyclic association schemes -- this appears to be the easiest case to treat. We are able to prove that the unitary entanglement-assisted independence number grows at least quadratically faster, with respect to matrix size, than the independence number introduced by Beigi and Shor as a generalization of the one-shot Shannon capacity, which naturally corresponds to the well-known graph-theoretic notion of independence number \cite{Beigi07}.

We recall the formal definition of an association scheme in the next section. The Kraus operators of quantum channels from symmetric association schemes commute. As a consequence, it is particularly easy to study properties of these channels. It is worth to remark that the link suggests a new set of parameters for association schemes. Quantum independence number and unitary entanglement-assisted independence number are discussed in sections 4 and 5, respectively. Our notation is the same as in \cite{Duan12}.

\section{Quantum channels from association schemes}

\subsection{Association schemes}
An \emph{association scheme} with $d$-classes is a set of $(0,1)$-matrices $\mathcal{A}=\{A_0,\cdots, A_d\}$ satisfying the following four properties: (1) $A_0=I$, where $I$ is the identity matrix; (2) $\sum_{i=0}^d A_i=J$, where $J$ is the all-ones matrix; (3) $A_i^T\in\mathcal{A}$, where $A_i^T$ is the transpose of $A_i$; (4) $A_iA_j\in span{\mathcal{A}}$. The $\C$-linear span of $A_0,A_1,\cdots,A_d$ forms a semisimple  algebra of dimension $d+1$, called the {\it Bose-Mesner algebra} of the scheme.
Since each $A_i$ is symmetric, this algebra is commutative. Given that this algebra is closed under complex conjugation and contains the identity, there exists a set of minimal idempotents $E_0,E_1,\cdots,E_d$  which also forms a basis of the algebra.  We have $E_i^\dag=E_i$ for each $i$, $E_iE_j=\delta_{i,j}E_i$. An association scheme is called \emph{symmetric} if $A_i^T=A_i$ for each $i$. A symmetric association scheme must be commutative.

In our examples, we will consider a special type of association schemes called pseudocyclic. A scheme is \emph{pseudocyclic} if $rank(E_i)$'s are all equal other than $rank(E_0)$, where $rank(E_0)=1$. Let us describe the cyclotomic association schemes, which are pseudocyclic. Let $q=p^f$, where $p$ is a prime and $f$ a positive integer. Let $\gamma$ be a fixed primitive element of $\F_q$ and $d|(q-1)$ with $d>1$. Let $C_0=\langle \gamma^d\rangle$, and $C_i=\gamma^i C_0$ for $1\leq i\leq d-1$. In particular, when $d=2$, $C_0$ is exactly the nonzero squares in $\F_q$.  Define a set of $q\times q$ matrices with rows and columns labeled by elements of $\F_q$ as follows: set $A_0=I$, and for each $i\in \{1,2,\ldots ,d\}$, define $A_i$ to have its $(x,y)$-th entry equal to $1$ if $x-y\in C_{i-1}$ and $0$ otherwise. Then $\{A_i\}_{0\leq i\leq d}$ is a pseudocyclic association scheme. A standard reference on association schemes is \cite{Bannai84}.

\subsection{Non-commutative graphs}
Let $\eL(\eH_A)$ and $\eL(\eH_B)$ be the spaces of linear operators on the finite dimensional Hilbert spaces $\eH_A$ and $\eH_B$. A quantum channel is a map which describes the dynamical evolution of states in a quantum system. Formally, a \emph{quantum channel} is a map $\eN:$ $\eL(\eH_A) \rightarrow \eL(\eH_B)$ that is complete positive and trace preserving (for short CPTP). A CPTP map is represented by a set of Kraus operators $F_i: \eH_A\rightarrow \eH_B$ such that $\eN(\rho)=\sum_iF_i\rho F_i^\dag$ and $\sum_i F_iF_i^\dag=I$. The \emph{non-commutative graph} associated with the channel $\eN$ is the operator subspace $S=span\{F_i^\dag F_j:\,i,j\}<\eL(\eH_A)$. Observe that $S=S^\dag$ and $I\in S$. A non-commutative graph is special type of operator system \cite{Duan12, Paul03}. A non-commutative graph represents the confusability relations between the input/output symbols of the quantum channel. It has the same role as a graph for a stationary memoryless channel in Shannon zero-error information theory \cite{Shannon56}. Indeed, for this type of channels, non-commutative graphs reduce to graphs.

\begin{definition}
The \emph{independence number} of $S$, denoted by $\al(S)$, is defined as the integer $\al(S)=\max|\{|\phi_i\rangle: i=1,\cdots,N\}|$, such that $\langle\phi_i|\phi_j\rangle=0$ and $|\phi_i\rangle\langle\phi_j|\in S^\perp$, $i\ne j$.
\end{definition}

The operator space $S^\perp$ is with respect to the Hilbert-Schmidt inner product on $\eL(\eH_A)$ defined as $\langle X,Y\rangle=\tr(XY^\dagger)$ with $X,Y\in \eL(\eH_A)$. One trivial observation here is that $\al(S)\leq dim_\C(\eH_A)$: if a set of nonzero vectors in $\eH_A$ are pairwise orthogonal, then they are linearly independent. Computing the number $\al(S)$ is QMA-complete \cite{Beigi07}; $\al(S)$ can be seen as a generalization of the independence number for graphs, whose computation is well-known to be NP-complete.

\subsection{Quantum channels}
We define a class of quantum channels from association schemes. Given an association scheme $\mathcal{A}$, the set of matrices $A_0,A_1,\cdots,A_d$ defines the Kraus operators of a CPTP map. We then define $S_{\mathcal{A}}=span\{A_i^\dag A_j:\,i,j\}$ . Clearly, we have $S_{\mathcal{A}}=span\{E_i:\,0\leq i\leq d\}$: this follows from the fact that $\mathcal{A}$ is an algebra closed under taking transpose, and the $E_i$'s form a basis.

\begin{thm}\label{a_as}
Let $\mathcal{A}$ be an association scheme with matrices of size $N\times N$. Then, $\al(S_{\mathcal{A}})=N$.
\end{thm}

\begin{proof}
Let $W_k$ be the image of the projection $E_k$. Assume that $|\phi_i\rangle$, $1\leq i\leq N$ form an ``independent set", namely $\langle\phi_i|\phi_j\rangle=0$ and $|\phi_i\rangle\langle\phi_j|\in S^\perp$ for each $i\ne j$. Decompose each $|\phi_k\rangle$ into the direct orthogonal sum $\phi_k=\sum_i \phi_{k,i}$, with $\phi_{k,i}=E_i|\phi_k\rangle \in W_i$. First, it is clear that $S_{\mathcal{A}}=span\{E_i:\,0\leq i\leq d\}$. Second, we look at the condition of $|\phi_i\rangle\langle\phi_j|\in S^\perp$: it amounts to $\tr(E_k|\phi_i\rangle\langle\phi_j|)=0$ for all $k$. Since $E_k^\dag=E_k=E_k^\dag E_k$, we have
\[
\tr(E_k|\phi_i\rangle\langle\phi_j|)=\tr(E_k|\phi_i\rangle\langle\phi_j|E_k^\dag)=\tr(|\phi_{i,k}\rangle\langle\phi_{j,k}|)=\langle\phi_{i,k}|\phi_{j,k}\rangle=0.
\]
Summing over $k$, this also guarantees that $\langle\phi_i|\phi_j\rangle=0$, since $I=\sum_i  E_i$.

Now we see that actually the nonzero vectors in the set $\{\phi_{k,i}: 1\leq k\leq N,0\leq i\leq d\}$ also satisfies the independence condition, and contain at least $N$ nonzero element. It is now routine to get such an independent set: take an orthonormal basis of each $W_k$ and take union. This shows the result.
\end{proof}

\noindent{\bf Remark:} This shows that the bound $\al(S_{\mathcal{A}})\leq N$ in \cite{Duan12} is tight. This observation justifies the study of quantum channels from association schemes, when attempting to understand what properties of a non-commutative graph are responsible for a separation between its classical capacity and capacities achievable with the aid of quantum resources.

\section{Channels with commuting Kraus operators}

We shall consider quantum channels satisfying the following hypothesis, which contains the symmetric association scheme case as a special case.\\

\noindent{\bf Hypothesis:}
\begin{enumerate}
\item The Hilbert space $\eH_A=\eH_B$, and they are finite dimensional over the  complex numbers $\C$, so that we regard $F_i$'s as matrices;
\item  The matrices $F_i$'s are normal matrices, i.e., $F_i^\dag F_i=F_i F_i^\dag$, so that each of them can be diagonalized;
\item The $F_i$'s commute, so that  they can be simultaneously diagonalized.
\end{enumerate}

\begin{lemma}\label{lem_h}
Under the above Hypothesis, the following properties hold:
\begin{enumerate}
\item  There exist pairwise orthogonal subspaces of $\eH_A$, denoted by $W_i$, $0\leq i\leq d$, and  numbers $\theta_k(i)$ such that
    \[
      F_k|v\ra=\sum_{i=0}^d\theta_k(i)|v\rangle,\;\forall\, k,\;\forall\, |v\rangle\in W_i.
    \]
    Moreover, for distinct $i,i'$, there exists $k$ such that $\theta_k(i)\neq\theta_k(i')$.
\item Denote by $E_i$ the projection onto $W_i$. We have  $E_i^\dag=E_i$, $E_iE_j=\delta_{i,j}E_i$, and $\sum_i E_i=I$.

\item  We have $F_k=\sum_{i=0}^m\theta_k(i)E_i$. Then $F_k^\dag=\sum_{i=0}^d\overline{\theta_k(i)}E_i$ , and
    \[
      F_k^\dag F_l=\sum_{i=0}^d\overline{\theta_k(i)}\theta_l(i)E_i.
    \]
\item $S:=span\{F_i^\dag F_j:\;i,j\}\leq span\{E_k:\;0\leq k\leq d\}.$
\end{enumerate}
\end{lemma}
\begin{proof} (1) follows from the properties of commuting normal matrices, as in the association scheme case \cite{Bannai84}. For (2), take an orthonormal basis $v_l$ of $W_i$, so that $E_i=\sum_l |v_l\rangle\langle v_l|$. It follows that $E_i^\dag=E_i$. (3) follows from (1), and (4) follows from (3).
\end{proof}

The following result generalizes Theorem \ref{a_as}.

\begin{thm}  Under the above Hypothesis, $\al(S) = dim_\C(\eH_A)$.
\end{thm}
\begin{proof} Adopt the same notation as in Lemma \ref{lem_h}.  We take an orthonormal basis of each $W_i$ and take union to form a set $T$. Clearly $T$ has size dim$_{\C}(\eH_A)$ and consists of pairwise orthogonal vectors. For any two distinct vectors $\phi_i$, $\phi_j$ in $T$ such that $\phi_i\in W_i$, $\phi_j\in W_j$, we have \[\langle\phi_j|F_k^\dag F_l|\phi_i\rangle=\overline{\theta_k(j)}\theta_{l}(i)\langle\phi_j|\phi_i\rangle=0,\]  so $|\phi_i\rangle\langle\phi_j|\in S^\perp$. Since $\al(S)\leq dim_\C(\eH_A)$, we see that the equality actually holds.
\end{proof}

\section{Quantum independence number}

Let us recall the definition of quantum independence number \cite{Duan12}.

\begin{definition}
The \emph{quantum independence number} of $S$, denoted by $\al_q(S)$, is defined as the maximum dimension of a subspace $A'\leq \eH_A$, with projection operator $P$, such that $PSP=\C P$.
\end{definition}

This means that there exists $\lambda_X\in \C$ such that $PXP=\lambda_X P$ for any $X\in S$. The quantum independence number corresponds to the largest dimension of a quantum error correcting code in the Knill-Laflamme setting \cite{KL97}.

\begin{thm}\label{aq_as}
 Let $\mathcal{A}$ be an $d$-class association scheme with minimal idempotents $E_i$, $0\leq i\leq d$. Then, $\al_q(S_\mathcal{A})=max\{rank(E_i):\,0\leq i\leq d\}$.
\end{thm}
\begin{proof} Take a nontrivial subspace $A'<A$ with projection operator $P$ such that $PSP=\C P$. We have seen that $S_\mathcal{A}$ is spanned by the minimal idempotent $E_i$'s. Therefore, there exist numbers $\lambda_i\in\C$ such that $PE_iP=\lambda_i P$ for each $0\leq i\leq d$. It is clear that $\lambda_i =1$ or $0$ since $P$ and all $E_i$'s are projections/idempotents, and so is their product  $PE_iP$.

Let $W_i$ be the image of $E_i$. Since $PE_iP=\lambda_i P$ and
$PE_iP$ has image $W_i\cap A'$, we see that $W_i\cap A'$ is either $\{0\}$ ($\lambda_i=0$) or $A'$ ($\lambda_i=1$). There is at least one $i$ such that $\lambda_i=1$; otherwise, $A'$ would be $\{0\}$. If  $\lambda_i=1$, then $W_i\cap A'=A'$, and $A'\leq W_i$.  Hence $A'$ is contained in one of the  $W_i$'s. We conclude that this quantum independent number does not exceed the largest multiplicity of the  association scheme. The equality is achieved if we take $A'=W_i$ whose dimension is maximal among all $W_i$'s. Observe that $rank(E_i)=dim_\C W_i$.\\
\end{proof}

\begin{thm} Suppose that our quantum channel satisfies the Hypothesis, and take the same notations as in Lemma \ref{lem_h}. Then the quantum independent number is equal to the maximal dimension of subspaces of $\eH_A$ on which  each $F_k^\dag F_l$ operates as a scale multiple.
\end{thm}
\begin{proof} Take a subspace $A'\leq A$ with projection operator $P$ such that $PSP=\C P$, and assume that it is not contained in another subspace with this property . Set $Y:=\{i:\, W_i\cap A'\neq \{0\}\}$. Then $PE_iP$ has image $W_i\cap A'$, so is not empty for each $i\in Y$. Since $F_k^\dag F_l\in S$, there exists a number $\lambda_{k,l}$ such that
\begin{align*}
\lambda_{k,l}P&=PF_k^\dag F_lP=\sum_{i=0}^d\overline{\theta_k(i)}\theta_l(i)PE_iP\\
  &=\sum_{i\in Y}\overline{\theta_k(i)}\theta_l(i)PE_iP.
\end{align*}
We see that for each $i\in Y$, the number $\overline{\theta_k(i)}\theta_l(i)=\lambda_{k,l}$ is a constant depending on $k,l$ only.
Set $W_Y:=\oplus_{i\in Y}W_i$ with projection operator $P_Y=\sum_{i\in Y}E_i$. Then
  \[
      F_k^\dag F_lP_Y=\sum_{i\in Y}\overline{\theta_k(i)}\theta_l(i)E_i=\lambda_{k,l}P_Y.
  \]
 Therefore,  all of $F_k^\dag F_l$'s  operate as a scale multiple of $P_Y$ on $W_Y$. Moreover, $P_YF_k^\dag F_lP_Y=\lambda_{k,l}P_Y$, so $W_Y$ also has the property that $P_YSP_Y=\C P_Y$. On the other hand, clearly $A'\leq W_Y$, since $\sum_{i}E_i=I$, and $P=\sum_iPE_iP=\sum_{i\in Y}PE_iP$. Since $A'$ is maximal, we get that $A'=W_Y$.

 Conversely, among all subspaces on which each $F_k^\dag F_l$ operates as a scalar multiple take $W$ to be one of maximal dimension. Then with $P$ the projection operator onto $W$, it is easy to see that $PSP=\C P$ and its dimension is the quantum independent number of $S$ by the above argument.
\end{proof}

\section{Unitary entanglement-assisted independence number}

 We consider in this section a further generalization of the independence number.

 \begin{definition}The \emph{unitary entanglement-assisted independence number}, $\widetilde{\alpha}_U(S)$, is the largest $N$ such that $\rho\in\mathcal{S}(\eH_A)\leq L(\eH_A)$ and unitary matrices $U_1,\ldots, U_N$ such that $U_m\rho U_{m'}\in S^\perp$ for any $m\neq m'$.
 \end{definition}

 Here $\mathcal{S}(\eH_A)$ consists of density operators: Hermitian positive semidefinite matrix with trace $1$. It was shown in \cite{Duan12} that $\al_q(S) \leq \al(S) \leq \widetilde{\alpha}_U(S)$. The integer $\widetilde{\alpha}_U(S)$ is defined when the encoding modulation of the channel is only unitary. If we extend the channel to a tensor product space (in the case of $\widetilde{\alpha}_U(S)$ the extension is trivially $\mathbb{C}$), then we can define the entanglement-assisted independence number. This quantity is motivated by the scenario where sender and receiver share an entangled state beforehand.

The next statement is a useful technical lemma. The result is known, but we recall the proof for the sake of completeness:

\begin{lemma} For each positive integer $t$, there exist $t^2$ unitary $t\times t$ matrices that are orthogonal w.r.t. the  Hilbert-Schmidt norm.
\end{lemma}
\begin{proof}
This is taken from \cite{Pittenger00}.  First, the adjusted basis $\{A_{j,k}:\,0\leq j,k\leq t\}$ is the set of $t\times t$ matrices defined by $A_{j,k}=E_{j,j+k}$, where the addition $+$ denotes addition modulo $t$. Then, we define the ``spin" matrices $S=\{S_{j,k}:0\leq j,k\leq t\}$ as follows:  $S_{j,k}=\sum_{r=0}^{t-1}F(j,r)A_{r,k}$ with $F(j,k)=exp(2\pi ijk/t)$. It is routine to show that these spin matrices are orthogonal w.r.t. the Hilbert-Schmidt norm.
\end{proof}

The proof technique of the following theorem is applicable to more general situations, but for simplicity and clarity, we just argue in the pseudocyclic association scheme setting.

\begin{thm} Suppose that our quantum channel satisfies the Hypothesis, and take the same notations as in Lemma \ref{lem_h}. Then, $\widetilde{\alpha}_U(S)\geq t^2d$ in the case $rank(E_0)=1$, $rank(E_1)=\cdots=rank(E_d)=t\geq 2$, where $D=dim_\C(\eH_A)$. In particular, $\widetilde{\alpha}_U(S)>\alpha(S)=1+td$.
\end{thm}
\begin{proof}
Define  $S_e:=\{E_0,\ldots, E_d\}$, so that $S\leq S_e$. Take $\rho\in\mathcal{S}(\eH_A)\leq \mathcal{L}(\eH_A)$  and unitary matrices $U_1,\ldots, U_N$ such that $U_m\rho U_{m'}\in S_e^\perp\leq S^\perp$ for any $m\neq m'$.

We first introduce two sets of orthonormal basis of $\eH_A$ as follows:
\begin{enumerate}
\item Take any orthonormal basis $w_1,\ldots, w_D$ of $\eH_A$, and define
$\rho:=\sum_{i=1}^t|w_i\rangle\langle w_i|\in\mathcal{S}(\eH_A)$. Let $W$ be the $D\times t$ matrix whose $i$-th column is $|w_i\rangle$, $1\leq i\leq t$.
\item For each $E_k$, we take an orthonormal basis $v_{k,1},\ldots, v_{k,n_k}$ of $W_k$ such that $E_k=\sum_{i=1}^{n_k}|v_{k,i}\rangle\langle v_{k,i}|$, where $n_k=rank(E_K)$.  Write $\{v_1,\ldots v_D\}:=\{v_{0,1},v_{0,2},\ldots, v_{d,n_d}\}$ in the natural order, and denote by $V$ the $D\times D$ matrix whose $i$-th row is the vector $|v_i\rangle$.
\end{enumerate}

The condition $U_m\rho U_{m'}\in S_e^\perp$ is equivalent to $\tr(U_{m'}^\dag\rho U_m E_k)=0$. With the above definition of $\rho$, this condition translates to $\sum_{l=1}^{n_k}\sum_{i=1}^t\langle w_i|U_{m}|v_{k,l}\rangle \langle v_{k,l}|U_{m'}^\dag|w_i\rangle=0$, or equivalently
\[
\sum_{l=1}^{n_k}\sum_{i=1}^t\langle w_i|U_{m}|v_{k,l}\rangle\overline{\langle w_i|U_{m'}|v_{k,l}\rangle}=0,\;\forall\; 0\leq k\leq d,\; m\neq m'.
\]

Now for each $m$, we define a $t\times D$ matrix $B_m:=W^\dag U_{m} V$. It is easy to check that $B_mB_m^\dag$ is the $t\times t$ identity matrix.
We write $B_m$ in the block matrix form $[B_m(0),B_m(1),\ldots,B_m(d)]$, where $B_m(k)$ is a $t\times n_k$ matrix. The above condition then translates to
\begin{equation}\label{b_cond}
\tr(B_{m'}(k)^\dag B_{m}(k))=0,\;\forall\; 0\leq k\leq d,\; m\neq m'.
\end{equation}

We take a maximal set  of unitary $t\times t$ matrices $\{C_{0},\ldots, C_{t^2-1}\}$ that are orthogonal w.r.t. the Hilbert-Schmidt norm. Now we construct a set of $t(D-1)=t^2d$ unitary matrices $U_m$ with the desired property. To do this, we first
define $t^2d$ of $t\times D$ matrices $B_m$'s as follows:  we set $B_{i+jt^2}(0)=0$, $B_{i+jt^2}(k)=C_i$ if $k=j$ and $=0$ otherwise, for each $0\leq i\leq t^2-1$, $1\leq j\leq d$. It is clear that Eqn. \eqref{b_cond} is satisfied. Now it is easy to find $U_m$ such that $B_m=W^\dag U_{m} V$: complete $B_m$'s and $W$ into unitary matrices and solve $U_m$. By our previous argument, these $U_m$'s have the desired property.
\end{proof}

\noindent{\bf Remark:} In the pseudocyclic association scheme case, we have $rank(E_0)=1$, and all other $rank(E_1)=\cdots=rank(E_d)=t$. Then the above theorem shows that  $\widetilde{\alpha}_U(S)\geq dt^2$, while the dimension of the ambient space is $D=1+dt$. If we fix $d$ and let $t$ tends to infinity, then asymptotically we shall have a sequence of quantum channels with $\frac{\widetilde{\alpha}_U(S)}{D^2}$ converging to some constant between $\frac{1}{d}$ and $1$. This can be realized by choosing proper cyclotomic schemes with a fixed $d$ as introduced at the beginning of this paper.\\

\noindent{\bf Remark:} We know that $\widetilde{\alpha}_U(S) \leq \widetilde{\alpha}(S)$. The above theorem gives then a separation between the operationally meaningful parameters  $\alpha(S)$ and $\widetilde{\alpha}(S)$.

\section{Open problems}

We conclude with some mathematical open problems:

\begin{enumerate}
\item For quantum channels from association schemes, can we give a better upper bound on $\widetilde{\alpha}_U(S)$ than the trivial bound $1+dim_\C S^\perp$ as in \cite{Duan12}?

\item For the pseudocyclic scheme with a fixed number of classes, describe the asymptotical behavior of the ratio $\frac{\widetilde{\alpha}_U(S)}{D^2}$, which we know is between $1/d$ and $1$, where $d$ is the number of classes.

\item What can we say about the entanglement-assisted capacity of the channels discussed in this paper apart from the lower bound $\widetilde{\alpha}_U(S)$? What can we say about the asymptotic behaviour of this capacity?

\item The independence numbers discussed here can be interpreted as new association schemes parameters. What are the properties of an association scheme $\mathcal{A}$ responsible for a separation between $\alpha(S_{\mathcal{A}})$ and $\widetilde{\alpha}_U(S_{\mathcal{A}})$. Classify association schemes on the basis of such properties.\\

\end{enumerate}

\noindent \emph{Acknowledgments.} Part of this work has been done while the authors were participating in the \emph{2012 Shanghai Conference on Algebraic Combinatorics}, held at Shanghai Jiao Tong University (SJTU). We would like to thank Eiichi Bannai for useful discussion and interest on the topic. SS is grateful to the Institute of Natural Sciences at SJTU for financial support and kind hospitality.

\end{document}